\documentclass[10pt,journal,comsoc]{IEEEtran}
\usepackage{graphicx} % Required for inserting images
\usepackage{graphics}
\usepackage{subfigure}
\usepackage{multirow}
\usepackage{makecell}
\usepackage{textcomp}
\usepackage{array}
\usepackage{xcolor}

\usepackage{tabularx}
\newcolumntype{Y}{>{\centering\arraybackslash}X} % 定義一個置中且自動換行的欄位格式 Y

\usepackage{tikz}
\usetikzlibrary{shapes.geometric, arrows, positioning, calc}

\usepackage{amsmath,amssymb,amsfonts}
\usepackage{bm}

\usepackage{algorithm}
\usepackage{algorithmic}
\usepackage[algo2e,linesnumbered,ruled]{algorithm2e}

\usepackage{amsthm}
\theoremstyle{definition}

\newtheorem{thm}{Theorem}

\SetKwProg{Fn}{Function}{}{end}

\SetAlgoSkip{}

\newcolumntype{M}[1]{>{\centering\arraybackslash}m{#1}}

\DeclareSymbolFont{symbolsC}{U}{txsyc}{m}{n}
\DeclareMathSymbol{\notniFromTxfonts}{\mathrel}{symbolsC}{61}

\title{SCOPE: Deterministic and Training-Free 3D UAV Deployment via Perimeter-based Heuristics}
\author{Chuan-Chi~Lai,~\IEEEmembership{Member,~IEEE}%,
    \IEEEcompsocitemizethanks{
        \IEEEcompsocthanksitem{This research was supported by the National Science and Technology Council, Taiwan, under Grant No. NSTC 114-2221-E-194-062-. 
		This work was also partially supported by the Advanced Institute of Manufacturing with High-tech Innovations (AIM-HI) from the Featured Areas Research Center Program within the framework of the Higher Education Sprout Project by the Ministry of Education (MOE) in Taiwan. \emph{(Corresponding author: Chuan-Chi~Lai.)}}
        \IEEEcompsocthanksitem{C.-C. Lai is with the Department of Communications Engineering, National Chung Cheng University, Minxiong Township, Chiayi County 621301, Taiwan, and also with the Advanced Institute of Manufacturing with High-tech Innovations (AIM-HI), National Chung Cheng University, Minxiong Township, Chiayi County 621301, Taiwan (e-mail: chuanclai@ccu.edu.tw).}
    }
}
\IEEEtitleabstractindextext{%
\begin{abstract}
Unmanned Aerial Vehicle (UAV) mounted Base Stations (UAV-BSs) provide flexible coverage for temporary hotspot scenarios; however, efficiently optimizing 3D deployment to satisfy heterogeneous user distributions remains a significant challenge. While Deep Reinforcement Learning (DRL) approaches have shown promise, they often suffer from prohibitive training overhead and poor generalization in cold-start scenarios where the user topology is unknown a priori. To address these limitations, this paper proposes Satisfaction-driven Coverage Optimization via Perimeter Extraction (SCOPE), which is a deterministic and training-free 3D deployment framework. Unlike existing heuristics that rely on fixed-altitude assumptions, SCOPE integrates a perimeter-based peeling strategy with the Welzl Smallest Enclosing Circle (SEC) algorithm to dynamically optimize 3D positions. Theoretically, we provide a rigorous convergence proof and derive a polynomial time complexity of $O(N^2 \log N)$, ensuring predictable execution for real-time applications. Experimentally, we evaluate SCOPE in unpredictable hotspot environments against both traditional heuristics and state-of-the-art DRL baselines under a matched hardware budget. Simulation results demonstrate that SCOPE maintains a high user satisfaction rate between 82\% and 88\% while generating solutions within millisecond-level latency on commodity hardware. Furthermore, SCOPE demonstrates exceptional resilience by maintaining an approximate 40\% functional coverage rate at a minimum altitude constraint of 60 m; in this challenging regime, baseline methods suffer a significant performance degradation, dropping to approximately 20\% due to altitude-induced path loss. These findings validate SCOPE as a robust and agile solution for establishing instantaneous digital lifelines in zero-day disaster response missions.
\end{abstract}

}

\begin{document}

\maketitle
\IEEEdisplaynontitleabstractindextext

% \fi
\IEEEpeerreviewmaketitle

\section{Introduction}
\label{sec:intro}

\IEEEPARstart{W}{ith} the rapid evolution of wireless communication technologies beyond 5G (B5G) and 6G, the demand for ubiquitous and high-capacity connectivity has surged exponentially~\cite{6736746}. While terrestrial base stations (BSs) serve as the backbone of cellular networks, their fixed deployment often lacks the flexibility to cope with sudden traffic spikes caused by temporary hotspots, such as large-scale outdoor gatherings or disaster relief operations. In this context, \textit{Unmanned Aerial Vehicles} (UAVs) equipped with communication payloads, also known as \textit{UAV-mounted Base Stations} (UAV-BSs) or \textit{Drone Small Cells} (DSCs), recognized as a key enabler for \textit{Non-Terrestrial Networks} (NTN) in the 6G era, have emerged as a promising solution~\cite{7470933,8918497,9275613}. Leveraging their high mobility, flexible 3D maneuverability, and superior \textit{Line-of-Sight} (LoS) probabilities, UAV-BSs can be dynamically deployed to enhance coverage~\cite{7417609,10529150}, \textit{Quality of Service} (QoS)~\cite{7510820,8642333,9177297}, and energy efficiency~\cite{10531095,JIANG202219} on demand. 

Despite the potential of UAV-assisted networks, the efficient 3D placement of UAV-BSs remains a critical optimization challenge, particularly in environments with highly heterogeneous user distributions. Traditional deployment strategies often simplify the problem by assuming uniform densities or constraining UAVs to a fixed flight altitude~\cite{alhourani2014optimal,7486987}. For instance, the conventional heuristic, \textit{Counter-Clockwise Spiral} (CCS), typically positions UAVs from the boundary inward utilizing a fixed coverage radius~\cite{7762053}. However, such rigid approaches fail to adapt to the irregular and clustered nature of real-world crowds, leading to suboptimal coverage and load imbalances.

Recently, data-driven approaches have evolved into highly complex architectures. For instance, architectures such as \textit{Graph Attention Networks} (GAT) enhanced \textit{Deep Reinforcement Learning} (DRL) have been proposed for joint UAV placement~\cite{11224636}, while Transformer-based models have been integrated to capture complex spatial-temporal dependencies~\cite{10654340}. 
%\color{blue}
While these learning-based methods excel in steady-state environments, they encounter significant convergence delays in cold-start scenarios. In the immediate aftermath of a disaster, the communication environment is characterized by extreme uncertainty and rapid topological shifts. While trained DRL models offer fast inference, they often suffer from the \textit{distribution shift} caused by altered user patterns, leading to severe policy degradation~\cite{lai2026stcl}. Adapting such models to an unseen disaster topology typically necessitates a fine-tuning process exceeding $10^6$ milliseconds of high-end computation, which introduces intolerable delays in time-critical rescue missions~\cite{10531095, lai2026stcl}. This motivates the design of a deterministic, training-free deployment framework that can establish an initial digital lifeline within milliseconds.

%\color{black}
In contrast, optimization-based approaches remain vital for their stability. Recent research has demonstrated the robustness of geometric placement optimization in multi-UAV networks~\cite{10423082}. For instance, geometric constraints such as the angle-of-radiation have been exploited to optimize 3D connectivity~\cite{11235946}, reinforcing that mathematical methods are still indispensable. While recent advances have incorporated capacity constraints into geometric partitioning, most existing methods lack a unified framework to jointly optimize 3D position, altitude, and user association under strict capacity limits.

%\color{blue}
To address these limitations, this paper proposes a deterministic and training-free 3D deployment framework named \textit{Satisfaction-driven Coverage Optimization via Perimeter Extraction} (SCOPE). Unlike data-driven models, SCOPE adopts a computational geometry approach to iteratively peel the deployment area. By integrating a perimeter extraction mechanism with the \textit{Smallest Enclosing Circle} (SEC) algorithm~\cite{welzl1991smallest}, SCOPE dynamically determines the optimal 3D position and altitude for each UAV-BS in polynomial time. This framework allows UAVs to fly at lower altitudes in dense clusters to mitigate interference, while ascending to higher altitudes in sparse areas to maximize coverage, effectively tailoring the deployment to the user distribution without requiring any prior training.

The main contributions of this paper are summarized as follows:
\begin{itemize}
    \item \textbf{SCOPE Framework}: We propose SCOPE, a training-free geometric heuristic for snapshot-based initial 3D deployment that jointly optimizes horizontal position and altitude. The algorithm utilizes an outside-in peeling strategy to prioritize boundary coverage, which ensures robust connectivity in heterogeneous environments.
    \item \textbf{Theoretical Analysis}: We provide a theoretical analysis of the framework, formally proving its convergence and deriving a polynomial time complexity of $O(N^2 \log N)$. This deterministic guarantee enables a fast-recalculation strategy for zero-day scenarios, contrasting with the stochastic instability and heavy training burden of DRL approaches.
    \item \textbf{Tri-Constraint Feasibility}: We introduce a capacity-aware pruning mechanism that strictly enforces three simultaneous constraints, which consist of individual backhaul capacity, the operational altitude floor, and minimum QoS requirements. This logic ensures that each deployed UAV-BS maintains a feasible and saturated load.
    \item \textbf{Performance Evaluation}: We conduct extensive simulations against geometric and learning-based baselines under a strictly matched hardware budget. Results demonstrate that SCOPE maintains a superior User Satisfaction Rate (USR) between 82\% and 88\% while generating solutions in approximately 185.35 ms. This performance lead confirms SCOPE's unique suitability for establishing robust digital lifelines in complex topographies where pre-training is unavailable and physical obstacles are prevalent.
\end{itemize}

\color{black}
The remainder of this paper is organized as follows. Section \ref{sec:related_work} reviews related work. Section \ref{sec:system_model} describes the system model and problem formulation. Section \ref{sec:algorithm} details the proposed SCOPE algorithm. Section \ref{sec:simulation} presents the simulation results, and Section \ref{sec:conclusion} concludes the paper.

\section{Related Work}
\label{sec:related_work}

The deployment of UAV-BSs has been extensively studied to enhance wireless coverage and capacity. Existing literature can be broadly categorized into static heuristics, numerical optimization, and learning-based strategies. Table \ref{tab:method_comparison} provides a qualitative comparison of these approaches.

\begin{table*}[!t]
%\color{blue}
\setlength{\tabcolsep}{5pt}
\caption{Comparison of UAV Deployment Approaches}
\label{tab:method_comparison}
\centering
\begin{tabular}{|l|c|c|c|c|}
\hline
\textbf{Feature} & \textbf{\makecell{Spiral/Voronoi \\ \cite{7762053,zhao2018research,10221731}}} & \textbf{\makecell{Meta-Heuristics \\ \cite{10529150,8642333,10361533}}} & 
\textbf{\makecell{SOTA DRL \\ \cite{11224636,10654340,10483540,2602.09994v1}}} & \textbf{SCOPE} \\
\hline
\hline
\textbf{Target Scenario} & Uniform low-density & Delay-tolerant offline planning & Steady-state continuous tracking & Zero-day cold-start emergency \\
\hline
\textbf{Capacity Awareness} & Blind (Pure geometric) & Iterative penalty & Reward-driven & Strict (Tri-constraint) \\
\hline
\textbf{Core Mechanism} & Geometry / Partitioning & Iterative Search & Neural Network & Perimeter Extraction \\
\hline
\textbf{Fleet Sizing} & Fixed (Pre-set) & Fixed (Pre-set) & Fixed (NN Constrained) & Dynamic (On-demand) \\
\hline
\textbf{3D Optimization} & Partial (Sequential) & Yes & Yes & Yes (Joint 3D) \\
\hline
\textbf{Training Cost} & None & None & High & None \\
\hline
\textbf{Time Complexity} & Low / Medium & High & High (Training) & Low ($O(N^2 \log N)$) \\
\hline
\textbf{Generalization} & High & Medium & Low (Re-training) & High (Training-free) \\
\hline
\textbf{Solution Type} & Deterministic / Iterative & Stochastic & Stochastic & Deterministic \\
\hline
\end{tabular}
\end{table*}

%\color{blue}
\subsection{Geometric and Static Heuristics}
Early research focused on establishing probabilistic \textit{Air-to-Ground} (AtG) channel models. It was demonstrated in~\cite{alhourani2014optimal} that an optimal UAV altitude exists to maximize the coverage radius. For multi-UAV scenarios, geometric patterns are commonly employed. The CCS deployment algorithm~\cite{7762053} places UAVs from the boundary inward. While computationally efficient, CCS assumes fixed coverage radii, leading to coverage voids in heterogeneous environments. Voronoi-based approaches~\cite{zhao2018research,10653059,10901041} partition the service area to improve fairness. Recently, advanced geometric methods such as capacity-constrained power diagrams have been proposed to explicitly address the load balancing issue by dynamically adjusting cell weights~\cite{10221731}. However, while these methods may adjust UAV altitudes to improve signal strength (SNR), they typically treat the optimization of each cell independently or neglect the aggregate inter-cell interference (SINR) in their objective function. This limitation can lead to aggressive altitude configurations that maximize local coverage but degrade global network performance due to severe co-channel interference. In contrast, SCOPE determines the UAV altitude based on the SEC, which inherently minimizes the coverage radius and the resulting interference footprint, ensuring strict QoS satisfaction in interference-limited 3D environments.

%\color{black}
\subsection{Numerical and Meta-Heuristic Optimization}
Traditional numerical methods often formulate UAV placement as \textit{Mixed-Integer Non-Linear Programming} (MINLP) problems. To address the complexity of MINLP, various iterative algorithms have been developed. For instance, a proximal stochastic gradient descent algorithm was recently proposed in~\cite{9763515} to jointly optimize fairness and energy efficiency in cellular networks. While such gradient-based methods are theoretically sound, they typically require careful hyperparameter tuning (e.g., step sizes) and suffer from slow convergence in large-scale networks. In contrast, heuristic algorithms like K-Means clustering and Voronoi tessellation~\cite{7486987} offer lower complexity but often rely on center-based assumptions that fail to cover boundary users effectively. Unlike these iterative or center-based approaches, SCOPE leverages a deterministic geometric construction (SEC) to guarantee optimal coverage radius without the need for gradient steps or extensive iterations.

Alternatively, meta-heuristic algorithms have been adopted to find near-optimal solutions. For instance, \textit{Particle Swarm Optimization} (PSO) was recently employed to optimize the energy-efficient deployment of VLC-enabled UAVs~\cite{10361533}. Similarly, a \textit{Genetic Algorithm} (GA) base solution was proposed to solve the 3D placement problem for maximizing communication coverage~\cite{10529150}. Additionally, a density-aware placement strategy based on GA were investigated in~\cite{8642333} to guarantee data rate requirements. 

While these meta-heuristics offer better flexibility than rigid geometric patterns and can handle non-differentiable objectives, they fundamentally suffer from slow convergence rates and are prone to getting trapped in local optima. As noted in~\cite{10529150}, the iterative nature of GA requires massive population updates, making it computationally expensive for real-time applications requiring millisecond-level response times. Parallel to AI, geometric optimization continues to evolve. For instance, a robust multi-UAV placement problem was recently formulated using iterative optimization algorithms to guarantee localization performance in cooperative networks~\cite{10423082}. This validates our motivation, as mathematical optimization often provides better theoretical guarantees and efficiency than stochastic search heuristics for static or snapshot placement problems.

%\color{blue}
\subsection{Learning-based Deployment Strategies}
With the advent of AI, data-driven methods have become a dominant paradigm, evolving from traditional DRL to highly complex architectures. While DRL agents can learn policies for joint trajectory and power control~\cite{10666852}, more sophisticated \textit{Multi-Agent Reinforcement Learning} (MARL) frameworks have been developed to enhance stability in dynamic environments. A representative example is the ORCHID framework~\cite{2602.09994v1}, which introduces a \textit{Reset-and-Finetune} (R\&F) mechanism within the \textit{Multi-Agent Proximal Policy Optimization} (MAPPO) architecture to suppress gradient variance while utilizing a \textit{Ground Base Station} (GBS)-aware topology partitioning strategy to mitigate the exploration cold-start problem.

While these advanced models are exceptionally powerful for continuous trajectory tracking and long-term network management, they face fundamental challenges in emergency, cold-start deployments where the post-disaster environment deviates significantly from pre-training datasets~\cite{10531095,Dulac2021Challenges}. Although a pre-trained DRL model offers rapid inference times, the drastic distribution shift in user topology and radio propagation caused by disasters often leads to catastrophic forgetting and severe policy degradation~\cite{lai2026stcl,2602.09994v1}. In such cases, adapting the model through fine-tuning or retraining is not only time-consuming but also lacks guaranteed convergence within the critical initial minutes of a rescue operation~\cite{10531095}. To complement these learning-based strategies, our proposed SCOPE framework adopts a training-free, demand-driven approach tailored for this immediate deployment phase. By utilizing computational geometry to derive snapshot deployment solutions directly without any historical dependency~\cite{10423082}, SCOPE provides both computational agility and deterministic performance, serving as a robust baseline for real-time missions where the overhead and uncertainty of retraining are unacceptable.

\section{System Model and Problem Formulation}
\label{sec:system_model}

We consider a downlink UAV-assisted wireless network deployed in a post-disaster emergency scenario where the terrestrial infrastructure is severely compromised. The system operates within a $D \times D$ target area. As depicted in Fig.~\ref{fig:system_model}, a fleet of UAV-BSs, bounded by a maximum available budget $K$ and denoted by the index set $\mathcal{K}=\{1, \dots, K\}$, is dynamically deployed to establish a communication lifeline for $N$ arbitrarily distributed ground users (GUs), denoted by the index set $\mathcal{G}=\{1, \dots, N\}$. To ensure reliable service, this aerial deployment must adapt to the heterogeneous user clusters while strictly adhering to the physical and system limitations, including the operational altitude boundaries (between $h_{\min}$ and $h_{\max}$), the minimum data rate requirement ($R_{\min}$) for each served user, and the total backhaul capacity constraint ($C_{\text{backhaul}}$) connected to the core network.

\subsection{Heterogeneous User Distribution and Geometric Configuration}
To address the snapshot-based initial deployment problem, we focus on the spatial optimization of the network. The 3D coordinate of the $j$-th UAV is denoted by $\mathbf{u}_j = (x_j, y_j, h_j)$, where $h_j$ is constrained between $h_{\min}$ and $h_{\max}$ to comply with aviation regulations. Let $\mathbf{U} = \{\mathbf{u}_1, \dots, \mathbf{u}_K\}$ denote the set of 3D coordinates for all deployed UAVs. The horizontal location of the $i$-th GU is denoted by $\mathbf{w}_i = (x_i, y_i, 0)$.

To capture the highly heterogeneous and irregular nature of crowds in emergency scenarios, such as survivors gathering along accessible paths or at scattered assembly points, we model the spatial distribution of GUs using a \textit{Spatial Branching Process} (SBP)~\cite{chiu2013stochastic}. Unlike traditional uniform models, the SBP generates hierarchical and non-linearly shaped clusters through a recursive stochastic mechanism. The process begins with a set of initial cluster centers, where each point acts as a parent node that generates a random number of offspring according to a Poisson distribution with mean $\lambda_{\mathrm{off}}$. These offspring are then spatially scattered around their respective parent nodes based on a 2D Gaussian distribution $\mathcal{N}(\mathbf{0}, \sigma_{\text{s}}^2 \mathbf{I})$, where $\sigma_{\text{s}}$ denotes the spatial dispersion parameter.

Crucially, in this recursive branching model, each newly generated offspring can further serve as a parent for the subsequent generation. This multi-generational propagation continues until the branching process terminates, resulting in complex, multi-centered, and often elongated cluster topologies that more accurately reflect the unpredictable distribution of users in a real-world disaster zone. The total number of users $N$ is obtained by sampling from the final realization of this branching process. This severe spatial heterogeneity often results in significant coverage voids and load imbalances, necessitating a robust 3D optimization framework to strategically position the UAV fleet.

\begin{figure}[!t]
    %\color{blue}
    \centering
    \includegraphics[width=.975\columnwidth]{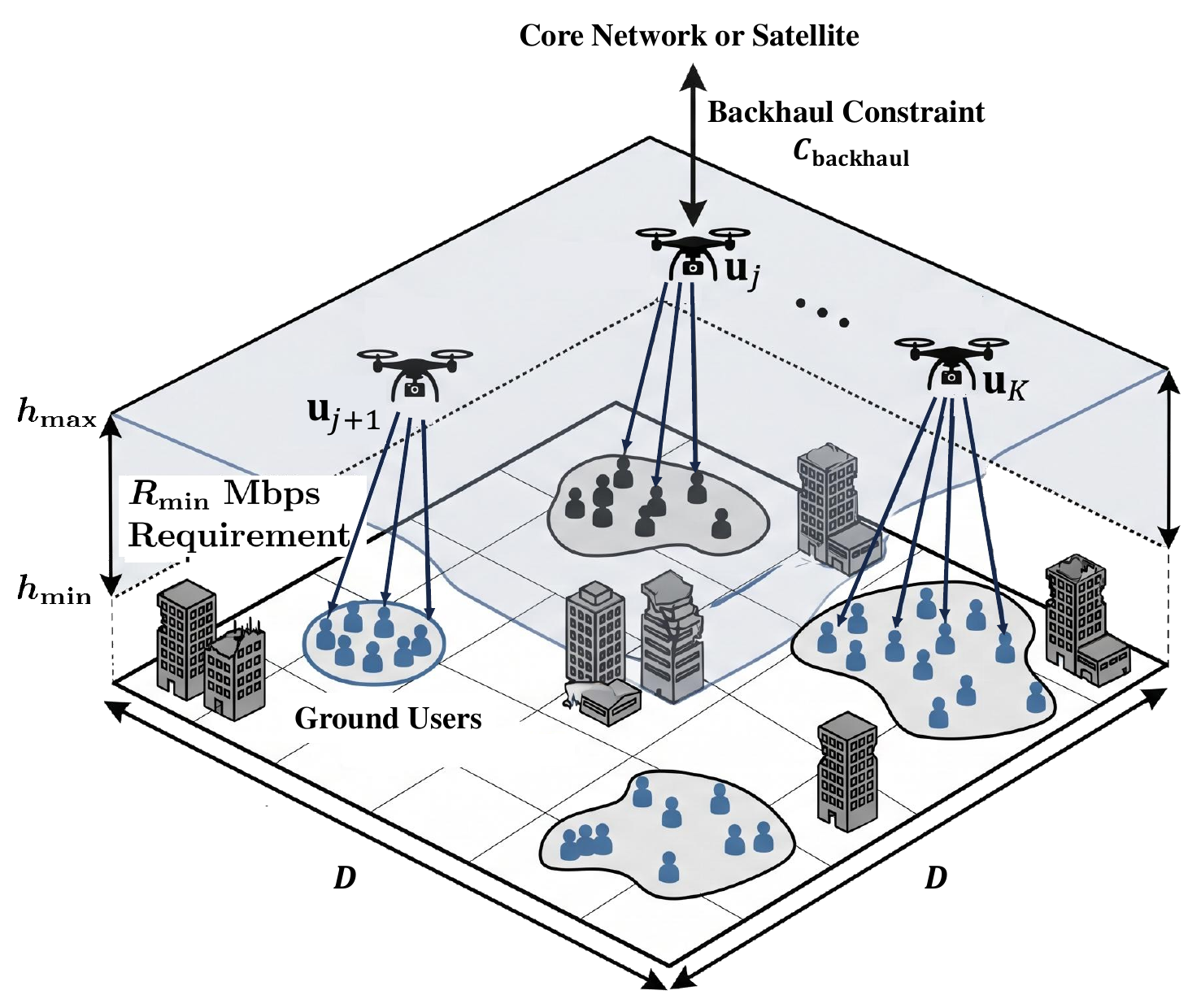}
    \caption{System model of the 3D UAV-BS deployment in a post-disaster scenario. The UAVs dynamically adjust their positions to serve heterogeneous ground user clusters while satisfying the operational altitude boundaries ($h_{\min}$ and $h_{\max}$), individual data rate requirement ($R_{\min}$), and the backhaul capacity constraint ($C_{\text{backhaul}}$) connected to the core network.}
    \label{fig:system_model}
\end{figure}

%\color{black}

\subsection{Air-to-Ground Channel Model}
The communication links are modeled using the probabilistic \textit{Air-to-Ground} (AtG) path loss model~\cite{alhourani2014optimal}. The probability of establishing a \textit{Line-of-Sight} (LoS) link between UAV $j$ and GU $i$ is given by the sigmoid function:
\begin{equation}
\label{eq:prob_los}
    P_{LoS}(\theta_{i,j}) = \frac{1}{1 + a \cdot \exp(-b(\theta_{i,j} - a))},
\end{equation}
where $a$ and $b$ are environment-dependent constants (e.g., urban, dense urban), $\theta_{i,j} = \frac{180}{\pi} \sin^{-1}\left(\frac{h_j}{d_{i,j}}\right)$ is the elevation angle, and $d_{i,j} = \|\mathbf{u}_j - \mathbf{w}_i\|$ representing the Euclidean distance. The average path loss is formulated as:
\begin{equation}
\label{eq:avg_pathloss}
    \bar{L}_{i,j} = P_{LoS} L_{LoS} + (1 - P_{LoS}) L_{NLoS},
\end{equation}
where the path loss for LoS and NLoS links are expressed as:
\begin{equation}
\label{eq:pathloss_components}
    L_{\xi} = 20\log\left(\frac{4\pi f_c d_{i,j}}{c}\right) + \eta_{\xi}, \quad \xi \in \{LoS, NLoS\},
\end{equation}
where $\eta_{\xi}$ representing the excessive attenuation factor.

%\color{blue}

\subsection{Communication and Interference Model}
We assume all UAV-BSs transmit with a constant power $P_t$. The received power at GU $i$ from UAV $j$ is $P_{r, i,j} = P_t \cdot 10^{-\bar{L}_{i,j}/10}$. In a high-density deployment scenario, co-channel interference from neighboring UAVs significantly impacts performance. The \textit{Signal-to-Interference-plus-Noise Ratio} (SINR) for GU $i$ associated with UAV $j$ is calculated as:
\begin{equation}
\label{eq:sinr}
    \Gamma_{i,j} = \frac{P_{r, i,j}}{\sum_{k \in \mathcal{U} \setminus \{j\}} P_{r, i,k} + \sigma^2},
\end{equation}
where $\sigma^2$ denotes the power of \textit{Additive White Gaussian Noise} (AWGN). The achievable data rate is then given by Shannon's capacity formula:
\begin{equation}
\label{eq:data_rate}
    R_{i,j} = B_{i,j} \log_2(1 + \Gamma_{i,j}),
\end{equation}
where $B_{i,j}$ is the bandwidth allocated to GU $i$.

\subsection{Power Consumption and Energy Efficiency}
Energy efficiency is a critical metric for battery-constrained UAVs. The total power consumption of a UAV-BS typically consists of communication power and propulsion (hovering) power. Since the mechanical power required to maintain altitude significantly dominates the circuit power, we adopt a simplified power model:
\begin{equation}
\label{eq:power_consumption}
    P_{\text{total}, j} \approx P_{\text{hover}} + P_{t},
\end{equation}
where $P_{t}$ is the transmission power and $P_{\text{hover}}$ is a constant representing the average hovering power consumption (e.g., based on rotor dynamics). Consequently, the system's \textit{Energy Efficiency} (EE) is defined as the total network throughput divided by the aggregate power consumption:
\begin{equation}
\label{eq:ee}
    \eta_{\text{EE}} = \frac{\sum_{j=1}^K \sum_{i=1}^N \mu_{i,j} R_{i,j}}{\sum_{j=1}^K (P_{\text{hover}} + P_{t}) \cdot \mathbb{I}(j)},
\end{equation}
where $\mathbb{I}(j)$ is an indicator function that equals 1 if the $j$-th UAV is active (deployed) and 0 otherwise. This formulation highlights that reducing the number of active UAVs is the most effective way to improve energy efficiency.

\subsection{Problem Formulation}
Let $\mathbf{M} = [\mu_{i,j}]_{N \times K}$ denote the user association matrix, where $\mu_{i,j} \in \{0,1\}$ is the binary association variable. Our objective is to maximize the User Satisfaction Rate (USR), $\mathcal{S}$, defined as the fraction of users whose data rate requirement $R_{\min}$ is met. We formulate the optimization problem for a given snapshot as follows:
\begin{subequations}
\label{eq:optimization_problem}
\begin{align}
    \max_{\mathbf{U}, \mathbf{M}} \quad & \mathcal{S} = \frac{1}{N} \sum_{i=1}^N \sum_{j=1}^K \mu_{i,j} \label{eq:objective_func}\\
    \text{s.t.} \quad & \sum_{j=1}^K \mu_{i,j} \le 1, \quad \forall i \in \mathcal{G}, \label{cons:association}\\
    & \sum_{i=1}^N \mu_{i,j} \le \gamma_{\max}, \quad \forall j \in \mathcal{K}, \label{cons:capacity}\\
    & h_{\min} \le h_j \le h_{\max}, \quad \forall j \in \mathcal{K}, \label{cons:altitude}\\
    & \mu_{i,j} R_{i,j} \ge \mu_{i,j} R_{\min}, \quad \forall i, j. \label{cons:qos}
\end{align}
\end{subequations}

The constraints of the proposed optimization problem are physically and operationally defined as follows. Constraint \eqref{cons:association} ensures that each ground user is associated with at most one UAV to avoid duplicate resource allocation. Constraint \eqref{cons:capacity} imposes a backhaul capacity limit on each UAV, where $\gamma_{\max} = \lfloor C_{\text{backhaul}} / R_{\min} \rfloor$ denotes the maximum number of users a single drone can simultaneously serve to prevent network congestion and hardware overload. Constraint \eqref{cons:altitude} restricts the flight altitude of each deployed UAV within the operational boundaries $h_{\min}$ and $h_{\max}$ to comply with aviation safety regulations. Finally, constraint \eqref{cons:qos} guarantees the required quality of service, ensuring that if user $i$ is associated with UAV $j$ (i.e., $\mu_{i,j} = 1$), the achievable data rate must meet or exceed the predefined threshold $R_{\min}$.

\subsection{Analysis of Problem Hardness}
\label{sec:hardness}
To justify the necessity of the proposed polynomial-time heuristic, we analyze the intractability of the MINLP problem formulated in \eqref{eq:optimization_problem} from three perspectives:
\begin{itemize}
    \item \textbf{Combinatorial Complexity:} The user association variable $\mu_{i,j}$ is binary. Even if the UAV locations $\mathbf{U}$ are fixed, the sub-problem of determining the optimal user association under capacity constraints \eqref{cons:capacity} and QoS constraints \eqref{cons:qos} resembles the \textit{Generalized Assignment Problem} (GAP), which is NP-hard. In the worst-case exhaustive search scenario, the potential solution space for associating $N$ users to $K$ UAVs scales as $O(K^N)$. For a dense urban scenario with $N=1000$ and $K=10$, the search space size becomes astronomical ($10^{1000}$), rendering brute-force approaches computationally intractable for real-time applications.
    \item \textbf{Non-Convexity of Continuous Optimization:} The continuous sub-problem of optimizing UAV coordinates $\mathbf{u}_j$ (for a fixed association $\mathbf{M}$) is non-convex. This arises primarily from the achievable data rate constraint \eqref{cons:qos}. The data rate $R_{i,j}$ defined in \eqref{eq:data_rate} is a logarithmic function of the SINR, which contains the Euclidean distance $d_{i,j} = \|\mathbf{u}_j - \mathbf{w}_i\|$ in the denominator of the path loss term. Specifically, the interference term $I_{\text{inter}} = \sum_{k \neq j} P_{r,i,k}$ in the denominator of the SINR expression \eqref{eq:sinr} makes the objective function and the feasible region highly non-convex. Consequently, standard convex optimization techniques (e.g., interior-point methods) cannot guarantee global optimality and may get trapped in local optima.
    \item \textbf{Variable Coupling:} The binary association variables $\mathbf{M}$ and continuous position variables $\mathbf{U}$ are intricately coupled in constraints \eqref{cons:qos} and \eqref{cons:capacity}. The optimal 3D position of a UAV depends on the set of users it serves (to minimize path loss), while the optimal set of served users depends on the UAV's position (to satisfy SINR thresholds). This mutual dependency prevents the problem from being decomposed into independent sub-problems. While iterative methods like \textit{Block Coordinate Descent} (BCD) can be applied, they typically suffer from slow convergence speeds and high computational overhead.    
\end{itemize}

In summary, the high complexity of the MINLP formulation prohibits the use of exact solvers (e.g., Branch-and-Bound) for online deployment. This theoretical bottleneck motivates the design of SCOPE, which relaxes the rigorous optimality for a deterministic, geometry-based solution with a worst-case complexity of only $O(N^2 \log N)$.

\color{black}

\subsection{Fairness Metric Definition}
In addition to maximizing user satisfaction, ensuring load balancing among UAV-BSs is crucial for preventing resource exhaustion and maintaining network stability. We verify the load fairness using \textit{Jain's Fairness Index} (JFI) based on the number of connections per UAV.

Let $\chi_j$ denote the load (number of served users) of the $j$-th UAV, which can be derived from the user association variable $\mu_{i,j}$:
\begin{equation}
\label{eq:uav_load}
    \chi_j = \sum_{i=1}^{N} \mu_{i,j}, \quad \forall j \in \mathcal{K}.
\end{equation}

The Jain's Fairness Index, denoted as $\mathcal{F}$, is defined as:
\begin{equation}
\label{eq:jain_fairness}
    \mathcal{F} = \frac{\left(\sum_{j=1}^{K} \chi_j\right)^2}{K \cdot \sum_{j=1}^{K} \chi_j^2}.
\end{equation}

The value of $\mathcal{F}$ ranges from $1/K$ to $1$. A value of $\mathcal{F} = 1$ indicates ideal fairness where all UAVs serve an equal number of users. Conversely, a lower value implies a disparity in load distribution, where some UAVs may be overloaded while others are underutilized.

\section{The Proposed SCOPE Framework}
\label{sec:algorithm}

In this section, we present the SCOPE (Satisfaction-driven Coverage Optimization via Perimeter Extraction) framework. The primary objective is to determine the optimal 3D coordinates $\mathbf{U}=\{\mathbf{u}_1, \dots, \mathbf{u}_K\}$
and the user association matrix $\mathbf{M}$
to maximize the user satisfaction rate while satisfying the backhaul capacity and geometric constraints defined in Section~\ref{sec:system_model}.

SCOPE operates on a peeling strategy derived from computational geometry. It iteratively identifies users on the boundary of the uncovered set, clusters them based on density constraints, and computes the optimal UAV altitude using the SEC algorithm.

\subsection{Algorithm Overview}
The flowchart of the SCOPE process is shown in Fig.~\ref{fig:flowchart}. The detailed steps, summarized in Algorithm~\ref{alg:scope_main}, consist of two nested phases that repeat until all users are served or available UAVs are exhausted:

\begin{enumerate}
    \item \textbf{Perimeter Extraction:} The algorithm iteratively identifies the boundary users (\textit{Convex Hull}) of the currently uncovered user set $\mathcal{G}_{\rm un}$. This peeling strategy ensures that the deployment prioritizes users at the edges who are most likely to be neglected by center-based heuristics.
    
    \item \textbf{Density-Adaptive Clustering:} Starting from a selected boundary user (seed), the algorithm constructs an initial maximal candidate set strictly bounded by a geometric search radius derived from the physical channel limits. 
    Instead of greedily aggregating neighbors, it adopts a \textit{top-down pruning} strategy. For the current candidate cluster, the SEC algorithm computes the minimal enclosing circle to determine the optimal horizontal position, while the altitude is derived to minimally cover this circle. The cluster feasibility is strictly constrained by three factors: 
    \begin{itemize}
        \item The backhaul capacity limit $\gamma_{\max}$ \eqref{cons:capacity},
        \item The maximum flight altitude $h_{\max}$ \eqref{cons:altitude}, and
        \item The QoS requirement $R_{\min}$ \eqref{cons:qos}, ensuring the worst-case user maintains valid connectivity.
    \end{itemize}
    If the candidate cluster violates any of these constraints, the algorithm iteratively removes the farthest user from the cluster center. This pruning process continues until all conditions are simultaneously satisfied, at which point the UAV is deployed to serve the finalized cluster.
\end{enumerate}

\begin{figure}[t]
\centering
\resizebox{\columnwidth}{!}{
\begin{tikzpicture}[node distance=1.2cm, auto,
    startstop/.style = {rectangle, rounded corners, minimum width=2cm, minimum height=0.8cm, text centered, align=center, draw=black, fill=gray!10, font=\footnotesize\bfseries, inner sep=2pt},
    process/.style = {rectangle, minimum width=2.2cm, minimum height=0.8cm, text centered, align=center, draw=black, fill=white, font=\footnotesize, inner sep=2pt},
    decision/.style = {diamond, aspect=2.2, minimum width=2.5cm, minimum height=0.8cm, text centered, align=center, draw=black, fill=white, font=\footnotesize, inner sep=0pt},
    io/.style = {trapezium, trapezium left angle=75, trapezium right angle=105, minimum width=2cm, minimum height=0.8cm, text centered, align=center, draw=black, fill=white, font=\footnotesize, inner sep=2pt},
    arrow/.style = {thick,->,>=stealth},
    line/.style = {thick,-}
]

\node (start) [startstop] {Start};
\node (input) [io, below of=start, node distance=1.1cm, text width=3cm] {Input: $\mathcal{G}$, QoS, Limits};
\node (init) [process, below of=input, node distance=1.1cm, text width=3cm] {Init: Mark all users\\as unserved};
\node (dec_outer) [decision, below of=init, node distance=1.6cm, text width=2.5cm] {Any users left?};
\node (hull) [process, below of=dec_outer, node distance=1.8cm, text width=4cm] {\textbf{Phase 1: Perimeter Extraction}\\Identify boundary users};
\node (init_cluster) [process, below of=hull, node distance=1.4cm, text width=3.8cm] {Pick seed \& form maximal\\candidate set $\mathcal{C}_{\text{test}}$};
\node (sec) [process, below of=init_cluster, node distance=1.4cm, text width=3.5cm] {\textbf{Phase 2: Pruning}\\Compute optimal 3D pos.};
\node (calc_h) [process, below of=sec, node distance=1.3cm, text width=3.5cm] {Calculate req. altitude\\\& worst-case QoS};
\node (dec_const) [decision, below of=calc_h, node distance=2.5cm, text width=2.8cm] {\textbf{All Limits Met?}\\1. Capacity?\\2. Altitude?\\3. QoS?};

\node (prune) [process, right of=sec, node distance=4.0cm, text width=2.5cm] {\textbf{Limits Violated!}\\Remove farthest user};
\node (deploy) [process, left of=sec, node distance=4.2cm, text width=2.8cm] {\textbf{All Satisfied!}\\Deploy UAV for current set};
\node (remove) [process, above of=deploy, node distance=2.2cm, text width=2.5cm] {Mark covered users\\as served};
\node (stop) [startstop, right of=dec_outer, node distance=4.0cm, text width=2cm] {Finish \& Output};

\draw [arrow] (start) -- (input);
\draw [arrow] (input) -- (init);
\draw [arrow] (init) -- (dec_outer);
\draw [arrow] (dec_outer) -- node[anchor=east, font=\footnotesize] {Yes} (hull);
\draw [arrow] (hull) -- (init_cluster);
\draw [arrow] (init_cluster) -- (sec);
\draw [arrow] (sec) -- (calc_h);
\draw [arrow] (calc_h) -- (dec_const);
\draw [arrow] (dec_outer) -- node[anchor=south, font=\footnotesize] {No} (stop);

% dec_const to prune (Loop back)
\draw [arrow] (dec_const.east) -| node[above, pos=0.25, yshift=0.5mm, font=\footnotesize] {No} (prune);
\draw [arrow] (prune) -- (sec);

% dec_const to deploy (Move forward)
\draw [arrow] (dec_const.west) -| node[above, pos=0.25, yshift=0.5mm, font=\footnotesize] {Yes} (deploy);
\draw [arrow] (deploy) -- (remove);
\draw [arrow] (remove) |- (dec_outer);

\end{tikzpicture}
}
\caption{Flowchart of the proposed SCOPE framework. The process iteratively peels the network from the boundary inward, constructing a maximal candidate set and iteratively pruning the farthest users until capacity, altitude, and QoS constraints are strictly satisfied.}
\label{fig:flowchart}
\end{figure}

\subsection{Perimeter Extraction and Initialization}
Let $\mathcal{G}_{\text{un}}$ denote the set of currently uncovered ground users, initialized as $\mathcal{G}_{\text{un}} \leftarrow \mathcal{G}$. In each iteration, SCOPE computes the boundary $\mathcal{B}$ of $\mathcal{G}_{\text{un}}$. We employ \textit{Andrew's Monotone Chain algorithm}~\cite{ANDREW1979216} to extract the convex hull. This choice ensures a strictly deterministic execution time regardless of user distribution. The boundary users are sorted in counter-clockwise order, i.e., $\mathcal{B}=\{b_1, b_2, \dots, b_m\}$, to facilitate orderly processing.

%----------------------------------------------------------------
% Algorithm 1: SCOPE Main Framework
%----------------------------------------------------------------
\begin{algorithm2e}[!t]
\SetAlgoLined
\caption{The SCOPE Deployment Framework}
\label{alg:scope_main}
\KwIn{Ground Users $\mathcal{G}$, Max UAVs $K$, Backhaul $C_{\text{backhaul}}$, Min Rate $R_{\min}$, Alt Limits $[h_{\min}, h_{\max}]$.}
\KwOut{Set of UAV-BSs $\mathbf{U}$, User Association $\mathbf{M}$.}

Initialize $\mathcal{G}_{\rm un} \leftarrow \mathcal{G}$, $\mathbf{U} \leftarrow \emptyset$, $\mathbf{M} \leftarrow \mathbf{0}$\;
Calculate capacity limit: $\gamma_{\max} \leftarrow \lfloor C_{\text{backhaul}} / R_{\min} \rfloor$\;
\While{$\mathcal{G}_{\rm un} \neq \emptyset$ \textbf{and} $|\mathbf{U}| < K$}{
    $\mathcal{B} \leftarrow \text{GetConvexHull}(\mathcal{G}_{\rm un})$\;
    Sort $\mathcal{B}$ in counter-clockwise order\;
    Select seed user: $b_{\rm seed} \leftarrow \mathcal{B}[0]$\;
    
    \tcp{Pass altitude limits to Algorithm 2}
    $(\mathbf{u}_{\rm opt}, \mathcal{C}_{\rm served}) \leftarrow \text{TopDownPruningSEC}(b_{\rm seed}, \mathcal{G}_{\rm un}, \gamma_{\max}, [h_{\min}, h_{\max}])$\;
    
    $\mathbf{U} \leftarrow \mathbf{U} \cup \{\mathbf{u}_{\rm opt}\}$\;
    
    \tcp{Update association matrix}
    Update $\mathbf{M}$ for users in $\mathcal{C}_{\rm served}$\;
    
    $\mathcal{G}_{\rm un} \leftarrow \mathcal{G}_{\rm un} \setminus \mathcal{C}_{\rm served}$\;
}
\Return{$\mathbf{U}, \mathbf{M}$}\;
\end{algorithm2e}

\subsection{Clustering and SEC Optimization}
For a selected boundary user $b_{\text{seed}} \in \mathcal{B}$, SCOPE attempts to form a service cluster $\mathcal{C}_{\text{served}}$ served by a UAV at location $\mathbf{u}_{\text{opt}}$. Mathematically, the formation of this cluster corresponds to determining the binary association variables for the $k$-th deployed UAV. Specifically, for all users $i \in \mathcal{C}_{\text{served}}$, we set $\mu_{i,k}=1$, and $\mu_{i,k}=0$ otherwise.

%\color{blue}
Unlike traditional greedy methods that build clusters bottom-up, SCOPE adopts a robust \textit{top-down pruning} strategy detailed in Algorithm~\ref{alg:sec_optimization}. To circumvent the $O(N^2)$ combinatorial explosion typical of unconstrained global searches, SCOPE bounds its search space using the physical channel models. Given the maximum allowable flight altitude $h_{\max}$ and the optimal elevation angle $\theta^*$ derived from the Air-to-Ground propagation environment~\cite{alhourani2014optimal}, the absolute maximum theoretical ground coverage radius is inherently bounded by $r_{\max} = h_{\max} / \tan(\theta^*)$. The process initializes a maximal candidate set $\mathcal{C}_{\text{test}}$ by strictly filtering all uncovered users within a geometric threshold of $2 \cdot r_{\max}$ centered around the seed user.

For the current candidate cluster, the optimal horizontal position $(x_c, y_c)$ and coverage radius $r_c$ are determined by solving the SEC problem~\cite{welzl1991smallest}. We employ \textit{Welzl's algorithm} to compute the unique minimal circle, which is geometrically determined by either two points (as a diameter) or three points (on the circumference) on the boundary of $\mathcal{C}_{\text{test}}$. This property ensures that the derived coverage radius $r_c$ is strictly minimal for the given set.

The required UAV altitude $h_{\text{req}}$ is then derived from the geometric constraint in~\eqref{cons:altitude}, while strictly enforcing the minimum flight altitude regulation:
\begin{equation}
    h_{\text{req}} = \max\left( h_{\min}, r_c \cdot \tan(\theta^*) \right).
\end{equation}
This lower-bound check guarantees that the UAV maintains a safe altitude $h_{\min}$ even when the cluster coverage radius $r_c$ is negligible.

If the current enclosing circle violates any constraints defined in the MINLP formulation, the algorithm iteratively prunes the set by removing the farthest user $g_{\text{far}}$ from the cluster center (excluding the seed). The pruning continues until all the following constraints are simultaneously satisfied:
\begin{itemize}
    \item The Capacity Constraint \eqref{cons:capacity}: $|\mathcal{C}_{\text{test}}| \le \gamma_{\max}$.
    \item The Altitude Constraint \eqref{cons:altitude}: $h_{\text{req}} \le h_{\max}$.
    \item The QoS Constraint \eqref{cons:qos}: The achievable data rate of the boundary user (worst-case in the cluster) must be greater than or equal to $R_{\min}$.
\end{itemize}
By utilizing this shrinking mechanism, SCOPE avoids the local optima often encountered in greedy expansions.
It guarantees that the derived solution captures the maximal feasible set of users within the geometric bound, while dynamic pruning ensures strict adherence to the complex tri-constraint mechanism.

\color{black}
%----------------------------------------------------------------
% Algorithm 2: Cluster Expansion
%----------------------------------------------------------------
\begin{algorithm2e}[!t]
\SetAlgoLined
%\color{blue}
\caption{Top-Down Pruning and SEC Optimization}
\label{alg:sec_optimization}
\KwIn{Seed $b_{\text{seed}}$, Candidates $\mathcal{G}_{\text{un}}$, Limit $\gamma_{\max}$, Min Rate $R_{\min}$, Alt Limits $[h_{\min}, h_{\max}]$, Opt. Elev. Angle $\theta^*$}.
\KwOut{Optimal Location $\mathbf{u}_{\text{opt}}$, Cluster $\mathcal{C}_{\text{served}}$.}

\tcp{Construct maximal initial set using optimal elevation angle}
Max Radius $r_{\max} \leftarrow h_{\max} / \tan(\theta^*)$\;
$\mathcal{C}_{\text{test}} \leftarrow \{ g \in \mathcal{G}_{\text{un}} \mid \text{dist}(g, b_{\text{seed}}) \le 2 \cdot r_{\max} \}$\;

\While{$|\mathcal{C}_{\text{test}}| \ge 1$}{
    $(x_c, y_c, r_c) \leftarrow \text{SEC}(\mathcal{C}_{\text{test}})$\;
    
    \tcp{Min altitude safety check based on optimal angle}
    $h_{\text{req}} \leftarrow \max(h_{\min}, r_c \cdot \tan(\theta^*))$\;
    
    \tcp{Check Geometric and QoS Feasibility}
    Calculate worst-case rate $R_{\rm edge}$ at height $h_{\rm req}$\;
    
    \eIf{$|\mathcal{C}_{\rm test}| \le \gamma_{\max}$ \textbf{and} $h_{\rm req} \le h_{\max}$ \textbf{and} $R_{\rm edge} \ge R_{\min}$}{
        $\mathbf{u}_{\text{opt}} \leftarrow (x_c, y_c, h_{\text{req}})$\;
        $\mathcal{C}_{\text{served}} \leftarrow \mathcal{C}_{\text{test}}$\;
        \textbf{break} \tcp*{All constraints satisfied}
    }{
        \tcp{Prune the worst/farthest user}
        Find $g_{\text{far}} \in \mathcal{C}_{\text{test}} \setminus \{b_{\text{seed}}\}$ farthest from $(x_c, y_c)$\;
        $\mathcal{C}_{\text{test}} \leftarrow \mathcal{C}_{\text{test}} \setminus \{g_{\text{far}}\}$\;
    }
}
\Return{$(\mathbf{u}_{\rm opt}, \mathcal{C}_{\rm served})$}\;
\end{algorithm2e}

\subsection{Theoretical Analysis}
In this subsection, we provide the theoretical convergence proof and complexity analysis of the proposed framework.

\begin{thm}[Convergence]
    The proposed SCOPE algorithm is guaranteed to terminate in a finite number of iterations, finding a valid deployment solution.
\end{thm}

\begin{IEEEproof}
Let $\mathcal{G}_{\text{un}}^{(k)}$ denote the set of uncovered users at the beginning of the $k$-th iteration. In each iteration, the algorithm selects a seed user $b_{\text{seed}} \in \mathcal{G}_{\text{un}}^{(k)}$ and forms a valid cluster $\mathcal{C}_{\text{served}}^{(k)}$. Since the top-down pruning step (Algorithm 2) strictly prevents the removal of the seed user, the algorithm guarantees that at least the seed user is served (i.e., $|\mathcal{C}_{\text{served}}^{(k)}| \ge 1$). Thus, the set of uncovered users strictly decreases:
\begin{equation}
    |\mathcal{G}_{\text{un}}^{(k+1)}| = |\mathcal{G}_{\text{un}}^{(k)}| - |\mathcal{C}_{\text{served}}^{(k)}| < |\mathcal{G}_{\text{un}}^{(k)}|.
\end{equation}
Given that the initial number of users $N$ is finite, the sequence $|\mathcal{G}_{\text{un}}^{(k)}|$ must reach 0 or the maximum number of UAVs $K$ is reached. Thus, the algorithm guarantees convergence.
\end{IEEEproof}

\begin{thm}[Time Complexity]
    The worst-case time complexity of SCOPE is strictly bounded by $O(N^2 \log N)$, avoiding the combinatorial explosion of unconstrained heuristics.
\end{thm}

\begin{IEEEproof}%\color{blue}
Let $N$ be the total number of GUs. SCOPE iteratively deploys UAVs to cover subsets of users. The computational complexity per deployed UAV is analyzed as follows:
\begin{enumerate}
    \item \textit{Perimeter Extraction}: Computing the convex hull of the currently uncovered users using Andrew's Monotone Chain algorithm requires $O(N \log N)$ time due to the initial lexicographical sorting step.
    \item \textit{Clustering and Pruning}: Unlike unconstrained greedy methods that evaluate all $N$ users, SCOPE restricts the search space using a physical geometric threshold $2r_{\max}$. Finding users within this bounded region takes $O(N)$ time. The subsequent distance calculation and sorting of this subset require $O(m \log m)$ time, where $m$ is the number of local candidates ($m \le N$). Once sorted, the top-down pruning process iteratively removes the farthest user in $O(1)$ time per operation, and is strictly capped by $\gamma_{\max}$. This ensures the entire pruning phase remains efficient. Crucially, because $m$ is physically bounded by the maximum geographic footprint $2r_{\max}$, the top-down pruning operations process a localized subset rather than the global population, thereby circumventing $O(N^2)$ inner loops.
    \item \textit{Smallest Enclosing Circle (SEC)}: Welzl's algorithm solves the SEC problem in expected linear time $O(m)$ for the candidate subset.
\end{enumerate}
In the extreme worst-case scenario where the spatial density is sparse and each deployed UAV covers only a minimal subset, the algorithm requires $K \approx O(N)$ sequential deployments. Summing the dominant convex hull extraction step across $K$ iterations yields a total bounded complexity of $\sum_{k=1}^{K} O(N_k \log N_k) \approx O(N^2 \log N)$. This confirms that SCOPE operates in low-order polynomial time, facilitating deterministic millisecond-level execution.
\end{IEEEproof}

%\color{blue}
\section{Simulation Results and Discussion}
\label{sec:simulation}
In this section, we conduct a comprehensive evaluation of the proposed SCOPE framework. All simulations are implemented in Matlab and executed on a desktop equipped with an AMD Ryzen 9 5950X CPU and 64 GB RAM. We compare SCOPE against multiple traditional heuristics and learning-based baselines to validate its effectiveness in terms of connectivity, equity, and resource efficiency.

\subsection{Simulation Setup and Baselines}
We consider a dynamic hotspot scenario within a $D \times D = 400 \times 400$ m$^2$ area. To simulate realistic extreme hotspot scenarios, the spatial positions of ground users are generated using a \textit{Spatial Branching Process} (SBP)~\cite{chiu2013stochastic}, which creates hierarchical and multi-centered user clusters. Unlike traditional static models, the SBP utilizes a recursive stochastic mechanism to generate offspring across multiple generations, accurately reflecting the unpredictable and often elongated distribution of survivors in post-disaster environments. To align with our focus on cold-start emergency missions, our simulation evaluates the algorithmic performance on these snapshot topologies to validate SCOPE's adaptability to heterogeneous densities without any prior knowledge. The channel parameters and physical constraints are summarized in Table~\ref{tab:params}.

In our experiments, we evaluate the performance using three primary metrics: USR ($\mathcal{S}$), JFI ($\mathcal{F}$), and EE ($\eta_{\text{EE}}$). While commercial cellular networks often prioritize the average communication rate to maximize individual experience, emergency post-disaster networks necessitate a fundamentally different approach. In such mission-critical scenarios, the primary objective is to establish a basic digital lifeline for the maximum number of survivors. Optimizing for the average data rate can inadvertently exacerbate near-far polarization, as algorithms may prioritize users directly beneath the UAV to artificially inflate the average metric while leaving boundary users in complete communication voids. Consequently, we select the USR as our primary objective. By strictly enforcing a minimum QoS threshold $R_{\min}$ for the boundary users during the cluster expansion phase, the proposed framework prevents this severe polarization and ensures equitable access to critical connectivity.

We compare SCOPE with the following baselines:
\begin{enumerate}
    \item \textbf{Counter-Clockwise Spiral (CCS)~\cite{7762053}:} A geometric heuristic that determines UAV placement by iteratively extracting the convex hull boundary and performing a counter-clockwise spiral search to cover ground users. It operates under a fixed altitude and predefined service radius, representing a deterministic static placement strategy.
    
    \item \textbf{K-Means Clustering:} A classic unsupervised learning method that minimizes the aggregate user-to-UAV distance but ignores coverage radius constraints. Critically, standard K-Means requires the cluster number $K$ to be manually predefined, introducing significant man-in-the-loop decision overhead that is impractical for on-demand emergency response. To establish a rigorous Oracle baseline and eliminate this decision overhead, we implement two variants:
    \begin{itemize}
        \item \textbf{KM ($K^{\text{SCOPE}}$):} The number of UAVs is set to the optimal $K$ dynamically derived by our proposed SCOPE framework.
        \item \textbf{KM ($K^{\text{CCS}}$):} The number of UAVs matches the output of the CCS heuristic.
    \end{itemize}
    This controlled dual-setup strictly isolates the algorithmic performance, demonstrating the superiority of our perimeter extraction strategy against traditional centroid-based clustering when both are given the exact same hardware budget.

    \item \textbf{Voronoi-based Deployment~\cite{zhao2018research}:} A geometric approach that divides the area based on dominance regions, often used for load balancing.
    
    \item \textbf{Random Deployment:} UAVs are randomly placed within the 3D airspace.
    
    \item \textbf{Deep Reinforcement Learning with Cold-Start:} To explicitly analyze the trade-off between deployment latency and solution optimality, we include \textit{Proximal Policy Optimization with Cold-Start} (PPO-CS) to benchmark the computational overhead and training costs in zero-day scenarios.
\end{enumerate}

\begin{table}[!t]
	%\color{blue}
\caption{Simulation Parameters}
\label{tab:params}
\centering
\begin{tabular}{|l|c|c|}
    \hline
    \textbf{Parameter} & \textbf{Symbol} & \textbf{Value} \\
    \hline
    \hline
    Carrier Frequency & $f_c$ & 2 GHz \\
    \hline
    System Bandwidth & $B$ & 20 MHz \\
    \hline
    Noise Power & $\sigma^2$ & $-101$ dBm \\
    \hline
    Transmit Power & $P_t$ & 0.1 W \\
    \hline
    Hovering Power & $P_{\text{hover}}$ & 150 W \\
    \hline
    Map Size & $D \times D$ & $400 \times 400$ m$^2$\\
    \hline
    Constrained UAV Altitude & $[h_{\min}, h_{\max}]$ & [20, 120] m\\
    \hline
    Urban Channel Env. Constants & $a, b$ & 12.08, 0.11 \\
    \hline
    Excess Attenuation (LoS / NLoS) & $\eta_{LoS}, \eta_{NLoS}$ & 1.6 dB, 23 dB \\
    \hline
    Min. Data Rate & $R_{\min}$ & 2 to 10 Mbps \\
    \hline
    Backhaul Capacity & $C_{\text{backhaul}}$ & 150 Mbps \\
    \hline
    Total Users & $N$ & 200 to 1000 \\
    \hline
    Max Users per UAV & $\gamma_{\max}$ & Dynamic \\
    \hline
    Spatial dispersion parameter & $\sigma_{\text{s}}$ & 11 m \\
    \hline
    Mean number of offspring & $\lambda_{\text{off}}$ & 0.9 \\
	\hline
	Optimal Elevation Angle~\cite{alhourani2014optimal} & $\theta^*$ & $54.62^\circ$ \\
    \hline
\end{tabular}
\end{table}

To ensure a fair and transparent comparison regarding energy efficiency and coverage, the flight altitude mechanisms for all evaluated algorithms are explicitly defined. The absolute operational boundaries are set to $h_{\min}$ of 20 m and $h_{\max}$ of 120 m as detailed in Table~\ref{tab:params}. The CCS heuristic and Random Deployment utilize a fixed typical altitude of 100 m. In contrast, SCOPE, PPO-CS, K-Means, and Voronoi dynamically adjust their altitudes based on the spatial distribution of their assigned user clusters to minimize the enclosing radius.

To ensure the statistical significance of our findings, all presented data points are averaged over 100 independent Monte Carlo simulations. In each run, the aforementioned Spatial Branching Process (SBP) is utilized to generate a completely new and highly heterogeneous spatial topology. This stochastic approach effectively emulates the complex clustering behavior of survivors in disaster zones, where individuals tend to aggregate in dense hotspots while sparse populations remain scattered at the network periphery. The 95\% confidence intervals are evaluated and depicted as error bars in the respective figures. By evaluating the algorithms across these 100 non-typical and highly uneven snapshots, we provide robust statistical evidence that the performance advantages of the deterministic SCOPE framework remain resilient against extreme environmental stochasticity.

\subsection{User Satisfaction Rate Analysis}
\label{sec:usr_analysis}

We first evaluate the user satisfaction rate $\mathcal{S}$ under varying user densities and QoS requirements. This metric serves as a direct indicator of the network reliability in disaster recovery missions.

\subsubsection{Varying number of users}

Fig.~\ref{fig:satisfaction:a} illustrates the satisfaction rate $\mathcal{S}$ against the total number of users $N$ (ranging from 200 to 1000) with a fixed $R_{\min}$ of 2 Mbps. As user density increases, the performance of static and unconstrained baselines degrades significantly. The CCS algorithm~\cite{7762053}, while being a classic heuristic, is limited by its predefined geometric spiral structure and fails to adapt to irregular user clusters. This leads to a relatively low $\mathcal{S}$ of approximately 43\% as $N$ reaches 1000. 

More severely, Voronoi and K-Means deployments suffer from a fundamental geometric flaw known as the QoS-blind partitioning bias. To isolate the impact of deployment logic, we introduced the $K^{\text{SCOPE}}$ and $K^{\text{CCS}}$ baselines, which utilize the exact same number of UAV-BSs as their respective counterparts. The results show that despite having the same hardware budget, K-Means variants perform poorly because they rely on centroid-based partitioning without capacity awareness. These methods frequently generate clusters that are too large to satisfy the tri-constraint feasibility. Consequently, the required service altitudes for these methods often hit the maximum ceiling $h_{\max}$, triggering fatal path loss and driving the satisfaction rate of K-Means below 2\%. Voronoi partitioning is even less effective, essentially failing to provide any valid coverage ($\mathcal{S} \approx 0\%$) in high-density scenarios.

Notably, the PPO-CS baseline performs poorly across all density levels, with satisfaction rates hovering below 10\% and even underperforming random deployment. This validates our assertion that untrained learning-based models produce highly inefficient 3D coordinates. In zero-day disaster scenarios where environmental priors are absent, this lack of convergence leads to a critical failure in establishing a reliable digital lifeline.

In contrast, the SCOPE framework demonstrates exceptional robustness, consistently maintaining a high satisfaction rate between 82\% and 88\%. This stability stems from the proposed top-down shrinking strategy. By targetting edge users through perimeter extraction and iteratively pruning the user topology inward, SCOPE strictly bounds the geometric radius to ensure excellent channel conditions. This deterministic mechanism effectively shields the network from the negative impacts of user densification.

\begin{figure}[!t]
    \centering
    \subfigure[Number of users $N$]{
        \label{fig:satisfaction:a}
        \includegraphics[width=0.9\columnwidth]{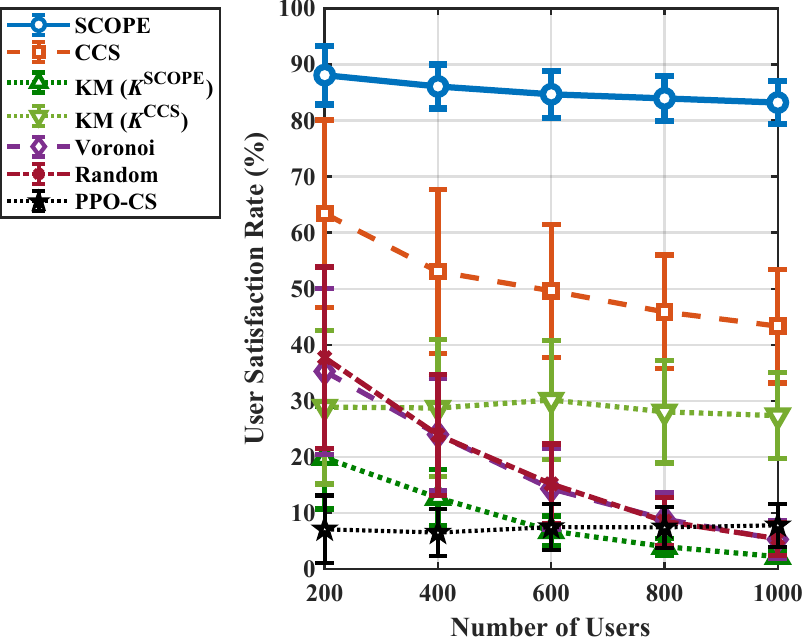}
    }\\%
    \subfigure[QoS requirement $R_{\min}$ (Mbps)]{
        \label{fig:satisfaction:b}
        \includegraphics[width=0.9\columnwidth]{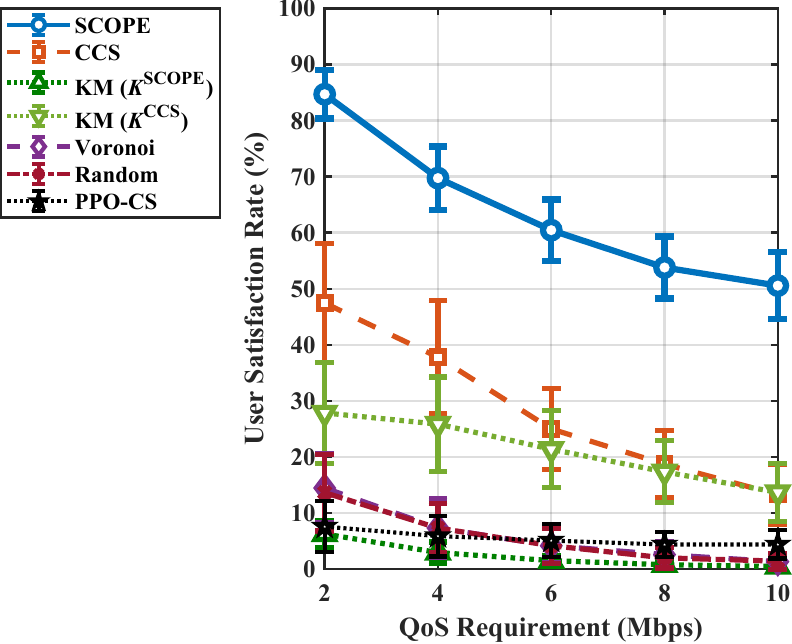}
    }
    \caption{Performance results in terms of user satisfaction rate $\mathcal{S}$ under varying
        \subref{fig:satisfaction:a} Number of users $N$ and 
        \subref{fig:satisfaction:b} QoS requirement $R_{\min}$ (Mbps).}%
    \label{fig:satisfaction}
\end{figure}

\subsubsection{Varying QoS requirements}

Fig.~\ref{fig:satisfaction:b} evaluates the satisfaction $\mathcal{S}$ under varying QoS requirements (2 to 10 Mbps) with $N=600$ users. As $R_{\min}$ increases, the individual bandwidth allocated per user effectively decreases, which imposes a tighter capacity bottleneck. 

SCOPE consistently maintains its performance lead, achieving a functional satisfaction rate of approximately 50\% even under the most stringent QoS threshold of 10 Mbps. This resilience is directly attributed to the tri-constraint mechanism and adaptive altitude adjustment via the SEC algorithm. By dynamically limiting cluster sizes through the pruning process, SCOPE prevents UAV altitudes from reaching inefficient heights, thereby minimizing path loss and maximizing the received SINR. In contrast, Voronoi and K-Means methods exhibit catastrophic performance drops. Their inability to verify QoS feasibility during placement forces UAVs to operate near $h_{\max}$ in a vain attempt to cover large areas, which is physically incompatible with high-capacity requirements. Furthermore, the PPO-CS baseline flatlines near 4\% to 7\%, reinforcing our conclusion that training-free geometric heuristics are mandatory for zero-day mission instantaneous deployment.

\subsection{Fairness Analysis}
\label{sec:fairness}

We verify the load balancing performance among deployed UAV-BSs using the JFI, denoted as $\mathcal{F}$. A higher $\mathcal{F}$ value indicates a more equitable distribution of served users across the active fleet, which is essential for preventing backhaul congestion in specific clusters.

\subsubsection{Varying number of users}

Fig.~\ref{fig:fairness:a} presents the fairness index against the number of users $N$ with a fixed QoS requirement of $R_{\min}=2$ Mbps. Across all density levels, SCOPE and CCS achieve the highest fairness indices, maintaining $\mathcal{F}$ between 0.20 and 0.30. 

The absolute values of the fairness index across all tested methods remain below 0.35. This phenomenon is a direct consequence of the highly heterogeneous spatial distribution generated by the SBP model. In realistic disaster scenarios, the distribution of users is inherently unbalanced; hotspots require certain UAVs to operate at full backhaul capacity, while the isolation of sparse groups at the network boundaries necessitates the deployment of UAVs with significantly lower loads. While this spatial disparity limits the mathematical fairness score, it accurately mirrors the operational reality of emergency response. In such missions, the priority is to maximize aggregate satisfaction rather than enforcing an artificial load balance across a non-uniform population.

Conversely, the K-Means variants and Voronoi methods demonstrate catastrophic load imbalance, with $\mathcal{F}$ values often dropping below 0.1. By partitioning the service area based on pure 2D geometry without enforcing capacity bounds during assignment, these methods assign massive crowds to a single UAV centroid while leaving other UAVs almost vacant. Furthermore, the untrained PPO-CS baseline flatlines at the bottom, proving that a neural network without converged environmental priors fails to distribute traffic loads geometrically.

\begin{figure}[!t]
    \centering
    \subfigure[Number of users $N$]{
        \label{fig:fairness:a}
        \includegraphics[width=0.9\columnwidth]{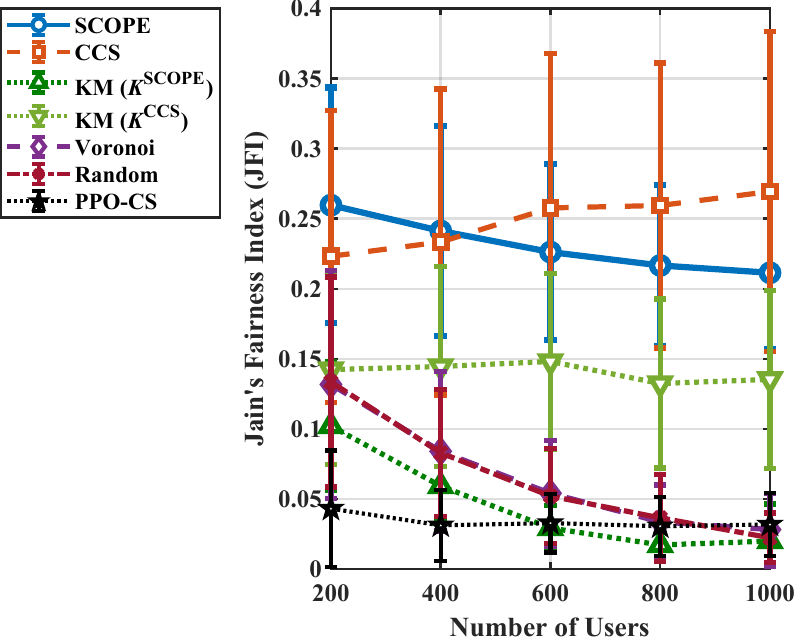}
    }\\%
    \subfigure[QoS requirement $R_{\min}$ (Mbps)]{
        \label{fig:fairness:b}
        \includegraphics[width=0.9\columnwidth]{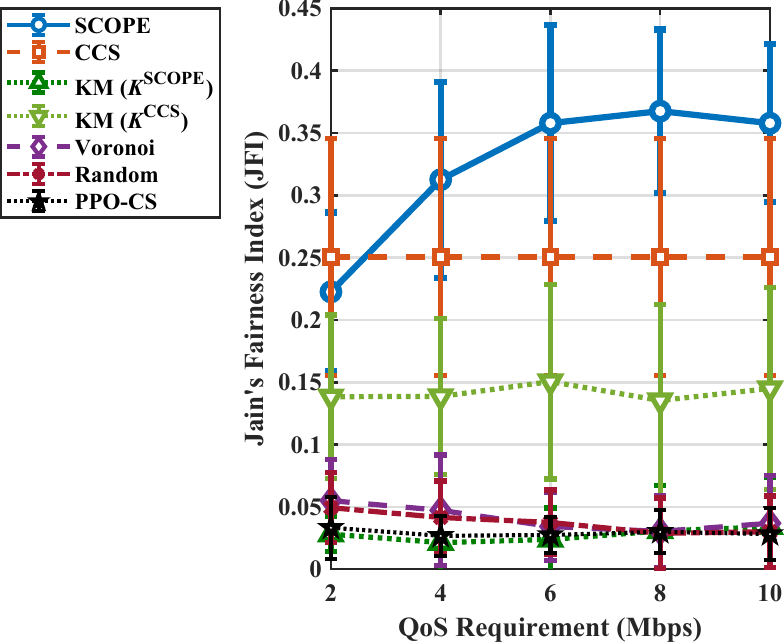}
    }
    \caption{Performance results in terms of fairness index $\mathcal{F}$ under varying
        \subref{fig:fairness:a} Number of users $N$ and 
        \subref{fig:fairness:b} QoS requirement $R_{\min}$ (Mbps).}%
    \label{fig:fairness}
\end{figure}

\subsubsection{Varying QoS requirements}

Fig.~\ref{fig:fairness:b} illustrates the fairness index $\mathcal{F}$ under varying QoS requirements with $N=600$ users. A fascinating and counter-intuitive phenomenon emerges as the minimum data rate becomes more stringent. While the baseline methods remain stagnant or degrade, the fairness index of SCOPE exhibits a significant upward trend, rising from approximately 0.22 at 2 Mbps and stabilizing near 0.35 under extreme constraints of 10 Mbps.

This positive correlation highlights the dynamic adaptability of the SCOPE tri-constraint pruning mechanism. As $R_{\min}$ increases, the effective capacity limit per UAV decreases because each user consumes a larger portion of the total backhaul. Consequently, SCOPE is mathematically forced to partition the topology into smaller and more uniformly sized clusters to ensure strict feasibility. This pruning-driven shrinking mechanism effectively breaks up massive density peaks and distributes the traffic load more evenly across a larger fleet of UAVs. 

In stark contrast, the CCS remains insensitive to QoS changes due to its reliance on a rigid static placement structure. Meanwhile, K-Means and PPO-CS remain paralyzed near the absolute bottom with $\mathcal{F} < 0.05$. This outcome validates that SCOPE is uniquely capable of actively restructuring its network topology to maintain equitable load balancing even under severe mission-critical bottlenecks.

\subsection{Energy Efficiency Analysis}

In this part, we analyze the energy efficiency (EE), denoted as $\eta_{\text{EE}}$, of the proposed SCOPE framework compared to the baseline algorithms. This metric is extremely critical in post-disaster scenarios where UAV-BSs are strictly constrained by limited onboard battery capacities and must justify every joule spent on hovering.

\subsubsection{Varying number of users}
Fig.~\ref{fig:ee:a} presents the energy efficiency performance against the number of users $N$ with $R_{\min}$ fixed at 2 Mbps.

The SCOPE framework consistently outperforms the evaluated benchmarks in all density scenarios, maintaining a high and stable energy efficiency between 38 and 40 Mbits/Joule. This robust performance gain is a direct mathematical result of the top-down pruning mechanism. Although SCOPE may deploy specific UAV-BSs to cover isolated boundary clusters to ensure high network reliability, the pruning process ensures that each stationary deployment is geometrically optimized to be capacity-saturated. This maximizes the valid aggregate throughput, which is the numerator of $\eta_{\text{EE}}$, relative to the 150 W hovering power expended per UAV. Consequently, SCOPE ensures that no electrical energy is wasted on maintaining sparse or low-quality links that do not satisfy the tri-constraint feasibility.

In contrast, the CCS algorithm demonstrates a moderate but consistently inferior performance. Although it utilizes a perimeter-based spiral placement strategy, its reliance on a rigid geometric structure often leads to deployments over areas with suboptimal user density. This structural rigidity prevents CCS from achieving the same capacity packing efficiency as SCOPE, resulting in wasted hovering power. More critically, the K-Means variants and Voronoi methods exhibit a severe degradation in EE as user density increases. Without a constraint-aware bounding mechanism, these methods form excessively large clusters that force UAVs to remain suspended at the altitude ceiling $h_{\max}$ to encompass boundary users. This triggers fatal path loss and massive connection failures, causing the valid throughput to plummet while the fleet continues to consume full hovering power. The untrained PPO-CS baseline similarly flatlines at the absolute bottom, proving that pseudo-random 3D deployments simply burn energy without establishing a valid digital lifeline.

\subsubsection{Varying QoS requirements}
Fig.~\ref{fig:ee:b} evaluates the energy efficiency under varying QoS requirements with $N=600$ users.

As $R_{\min}$ increases, the energy efficiency of SCOPE exhibits a graceful and highly resilient degradation, maintaining approximately 20 Mbits/Joule even at a stringent threshold of 10 Mbps. This stability stems from the tri-constraint pruning logic. As individual capacity demands increase, SCOPE dynamically reduces its coverage radii to maintain high-quality SINR, ensuring that the hovering power is always translated into valid throughput that meets the mandatory QoS level.

Conversely, the energy efficiency of the K-Means variants and Voronoi partitioning suffers a significant performance degradation as $R_{\min}$ tightens. Because these methods are unaware of the tightening QoS boundaries during the placement phase, their oversized clusters fail to support high data rate requirements, effectively disconnecting the vast majority of users. Mathematically, the effective throughput drops toward zero while the hovering power consumption remains constant, which annihilates the system energy efficiency. The persistent flatline of PPO-CS further corroborates that the deterministic, constraint-aware logic of SCOPE is indispensable for maximizing green communications in unpredictable and resource-constrained disaster networks.

\begin{figure}[!t]
    \centering
    \subfigure[Number of users $N$]{
        \label{fig:ee:a}
        \includegraphics[width=0.9\columnwidth]{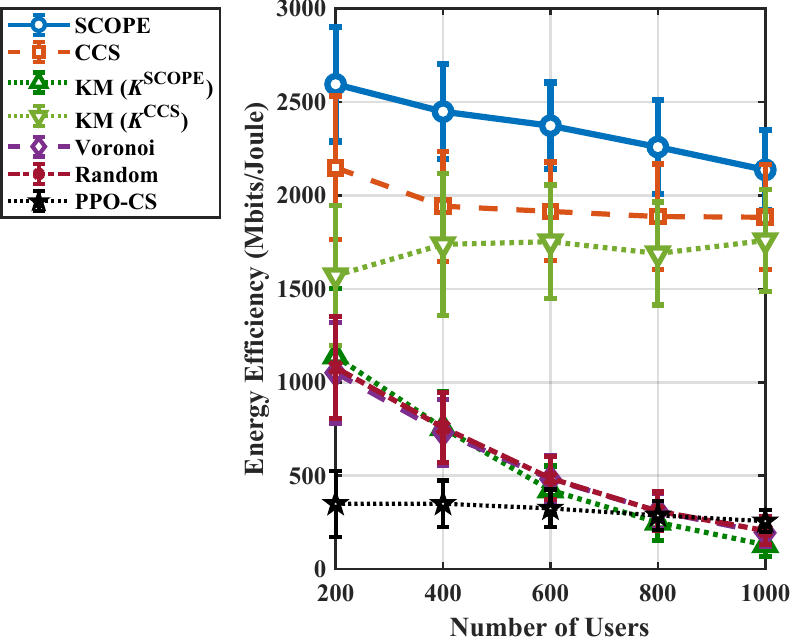}
    }\\
    \subfigure[QoS requirement $R_{\min}$ (Mbps)]{
        \label{fig:ee:b}
        \includegraphics[width=0.9\columnwidth]{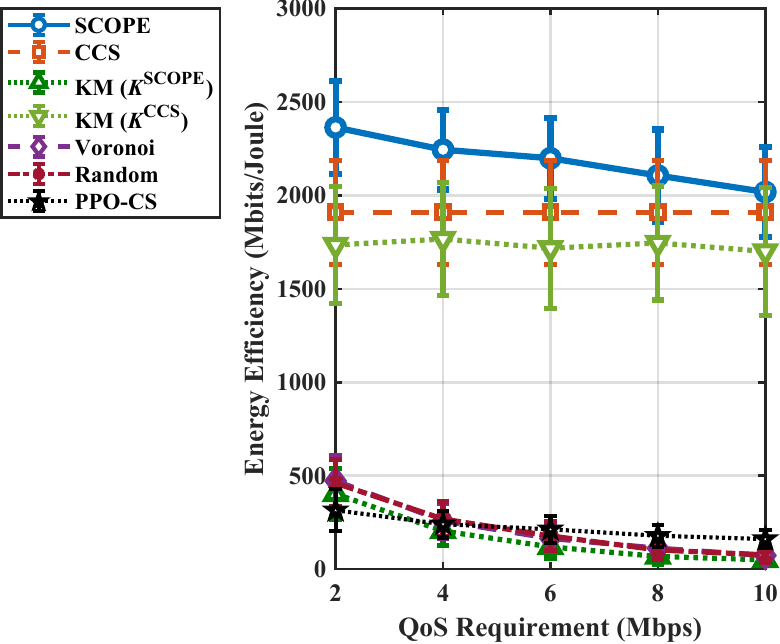}
    }
    \caption{Performance results in terms of energy efficiency $\eta_{\text{EE}}$ under varying
    \subref{fig:ee:a} Number of users $N$ and
    \subref{fig:ee:b} QoS requirement $R_{\min}$ (Mbps).}%
    \label{fig:ee}
\end{figure}

\subsection{Sensitivity Analysis}
\label{sec:sensitivity}
In this subsection, we evaluate the robustness of the SCOPE framework against varying hardware constraints and regulatory requirements. To isolate the source of performance gains as suggested by the reviewers, we conduct a controlled analysis by varying a single system parameter while maintaining a fixed user topology across 100 independent Monte Carlo runs.

\begin{figure}[!t]
    \centering
    \subfigure[Individual backhaul capacity $\gamma_{\max}$ (Mbps)]{
        \label{fig:sensitivity:a}
        \includegraphics[width=0.9\columnwidth]{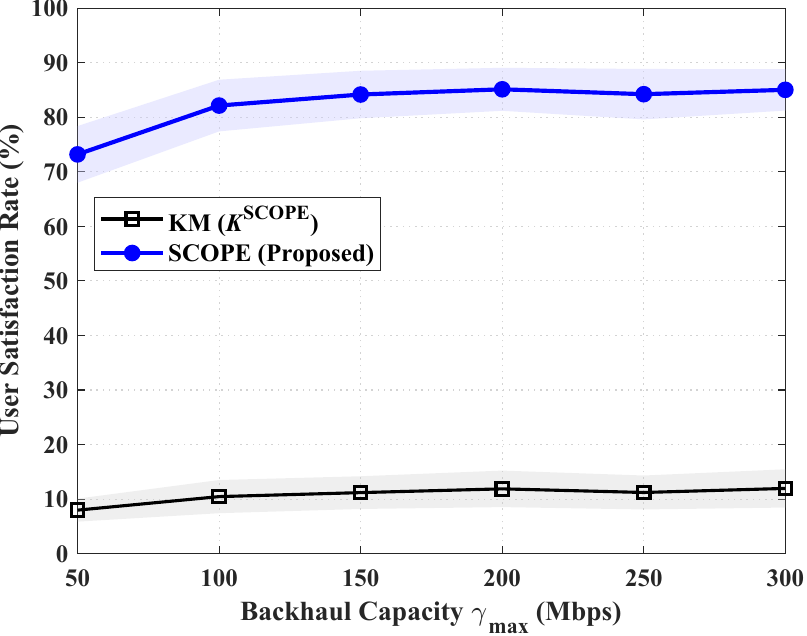}
    }\\
    \subfigure[Minimum altitude constraint $h_{\min}$ (m)]{
        \label{fig:sensitivity:b}
        \includegraphics[width=0.9\columnwidth]{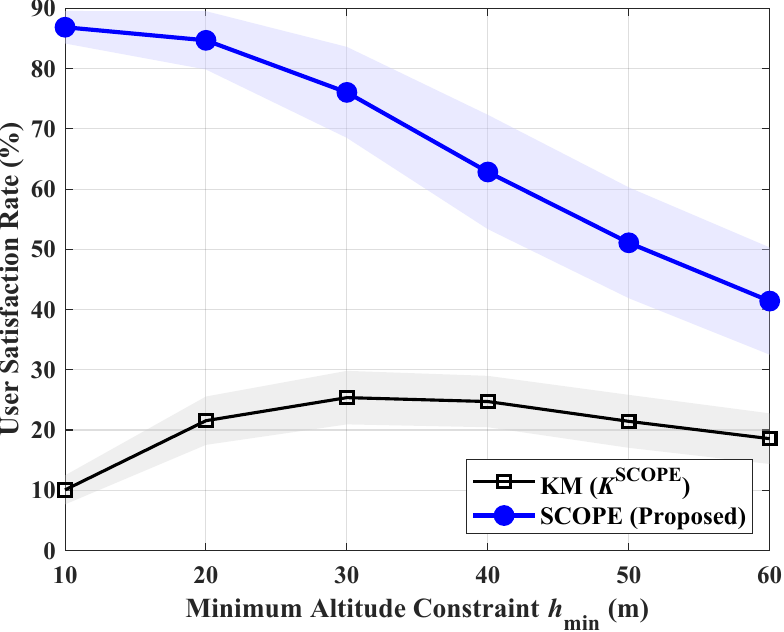}
    }
    \caption{System performance in terms of User Satisfaction Rate (USR) under varying 
        \subref{fig:sensitivity:a} individual backhaul capacity $\gamma_{\max}$ (Mbps) and 
        \subref{fig:sensitivity:b} minimum altitude constraint $h_{\min}$ (m).}%
    \label{fig:sensitivity}
\end{figure}

\subsubsection{Impact of Backhaul Capacity}
Fig.~\ref{fig:sensitivity:a} compares the performance of SCOPE against the KM ($K^{\text{SCOPE}}$) under varying individual UAV backhaul capacities, denoted as $\gamma_{\max}$, ranging from 50 Mbps to 300 Mbps. To maintain analytical focus and ensure a rigorous evaluation, other traditional heuristics and learning-based baselines are omitted from this specific sensitivity analysis. Specifically, traditional geometric methods lack native mathematical mechanisms to dynamically adapt their spatial placement to varying backhaul limits, making direct comparisons uninformative. Furthermore, learning-based approaches would require prohibitive retraining for each distinct capacity threshold to function properly. Consequently, the enhanced KM ($K^{\text{SCOPE}}$) serves as the most competitive and robust baseline. To ensure a strictly fair comparison, the KM ($K^{\text{SCOPE}}$) is configured to utilize the exact same number of UAV-BSs as dynamically determined by SCOPE for each simulation trial. The shaded regions represent the standard deviation across all trials.

The results demonstrate that SCOPE consistently maintains a higher USR across all backhaul conditions. While the performance of both algorithms improves as $\gamma_{\max}$ increases, KM ($K^{\text{SCOPE}}$) remains significantly less efficient, especially in bandwidth-constrained scenarios below 150 Mbps. This performance gap is attributed to the fact that KM ($K^{\text{SCOPE}}$) is a QoS-blind geometric partitioning method, which often assigns an excessive number of users to a single UAV regardless of backhaul limits. In contrast, SCOPE's capacity-aware pruning mechanism ensures that each UAV serves a feasible number of users to satisfy the tri-constraint requirements. The narrow standard deviation of SCOPE further confirms its deterministic stability, proving that the proposed framework maximizes network utility through intelligent resource-aware placement rather than a mere increase in hardware resources.

\subsubsection{Impact of Minimum Altitude Constraints}
Fig.~\ref{fig:sensitivity:b} evaluates the satisfaction rate sensitivity toward the minimum altitude constraint, $h_{\min}$, which represents a critical operational floor often dictated by urban obstacles or aviation safety regulations. To ensure a fair comparison, the KM ($K^{\text{SCOPE}}$) is enhanced to reverse-calculate its altitude based on the maximum cluster radius and is restricted to the same UAV budget as SCOPE. Note that traditional 2D heuristics (e.g., Voronoi and CCS) are excluded from this specific altitude sensitivity analysis, as they lack native mathematical mechanisms to dynamically adapt their vertical positioning to varying $h_{\min}$ constraints. This highlights SCOPE's distinct advantage in maintaining robust 3D coverage under stringent physical limitations.

The results indicate that as $h_{\min}$ increases from 10 m to 60 m, both algorithms exhibit a decline in USR due to the unavoidable physical path loss and the forced expansion of the minimum geometric footprint. However, SCOPE demonstrates remarkable resilience compared to the KM ($K^{\text{SCOPE}}$). When $h_{\min}$ reaches 60 m, the performance of KM ($K^{\text{SCOPE}}$) drops to approximately 20\% because it lacks the capacity to prune users who cannot satisfy the link budget at such altitudes. In contrast, SCOPE maintains a functional USR of approximately 40\%. This two-fold advantage is achieved through the top-down shrinking strategy, which strategically identifies and removes boundary users to protect the connectivity quality of the remaining cluster members. This evaluation demonstrates that SCOPE serves as a highly resilient solution for establishing reliable connectivity in highly obstructed urban topographies where low-altitude deployment is physically restricted.

\subsection{Computational Latency and Complexity Analysis}

A critical advantage of the proposed SCOPE framework over learning-based approaches is its exceptional responsiveness in zero-day scenarios. Table~\ref{tab:latency} summarizes the average execution time required to generate a complete stationary 3D deployment solution for a snapshot of $N=600$ users, calculated over 100 independent Monte Carlo trials.

As demonstrated in Table~\ref{tab:latency}, SCOPE generates a feasible deployment solution in approximately 185.35 ms on a standard single-core CPU. Although this latency represents a marginal increase compared to simpler heuristics such as CCS (65.55 ms), K-Means (4.35 ms), or Voronoi partitioning (5.43 ms), it remains well within the strict requirements for real-time emergency response. The computational overhead of SCOPE is a justified and minimal trade-off for superior connectivity performance. This delay is primarily attributed to the proposed top-down shrinking strategy, which iteratively prunes the user topology from the perimeter inward to ensure that the tri-constraint feasibility is strictly satisfied. Compared to QoS-blind methods, SCOPE ensures that a robust digital lifeline is established at the cost of less than 0.2 seconds of computation.

In contrast, DRL methods face a fundamental latency wall in disaster recovery. The PPO-CS baseline highlights a critical performance-latency dilemma; while the inference time of the model is exceptionally fast at 0.41 ms, the resulting deployment is often functionally unreliable due to the total lack of environmental priors. To achieve performance comparable to SCOPE, a DRL model must undergo an extensive training phase requiring over $10^6$ ms (approximately 17 minutes) of high-end GPU computation. In a zero-day disaster scenario where the user topology is unknown and deployment must be instantaneous, such a training overhead is prohibitive. 

Furthermore, in contrast to DRL models that approach environmental optimization as a black-box problem, SCOPE utilizes deterministic geometric priors; specifically, it leverages convex hull boundaries to prune the network topology. This inherent geometric awareness ensures feasibility under strict QoS constraints, thereby preventing the significant connection failures observed in Section~\ref{sec:usr_analysis}. In this context, the 185.35 ms execution time represents a justified trade-off for ensuring deterministic feasibility and stable connectivity under stringent mission constraints.

Beyond computational efficiency, this execution latency facilitates the effective management of user mobility. By achieving such high-frequency recalculation, the framework effectively ``freezes'' the network topology during each optimization window. For instance, given a ground user moving at a typical pedestrian speed of 1.2 m/s, the total spatial displacement during the 185.35 ms calculation period is approximately 0.22 m. This displacement remains negligible compared to the 120 m operational altitude and the resulting multi-meter coverage radius. Therefore, SCOPE can model the dynamic environment as a sequence of high-fidelity stationary snapshots without suffering from stale-state inaccuracies. SCOPE effectively bypasses the training limitations of DRL by providing a high-fidelity deployment solution in less than a quarter of a second; this is achieved without requiring prior environmental knowledge or specialized hardware acceleration.

\begin{table}[!t]
    %\color{blue}
\caption{Average Execution Time Comparison ($N=600$)}
\label{tab:latency}
\centering
\begin{tabular}{|l|c|c|}
\hline
\textbf{Algorithm} & \textbf{Time (ms)} & \textbf{Hardware Requirement} \\
\hline
\hline
\textbf{SCOPE (Proposed)} & \textbf{185.35 ms} & \textbf{CPU (Single Core)} \\
\hline
CCS (Spiral Search) & 65.55 ms & CPU (Single Core) \\
\hline
K-Means Clustering & 4.35 ms & CPU (Single Core) \\
\hline
Voronoi Partitioning & 5.43 ms & CPU (Single Core) \\
\hline
MINLP (Optimal Solver) & $> 10,000$ ms & CPU (Multi-Core) \\
\hline
PPO-CS (Inference) & 0.41 ms & CPU / NPU \\
\hline
PPO-CS (Training) & $> 10^6$ ms & High-end GPU \\
\hline
\end{tabular}
\end{table}

\subsection{Ablation Analysis of Design Components}
\label{sec:ablation}

To explicitly demonstrate the individual contribution of each core component within the SCOPE framework, we analyze the performance gains by mapping our specific architectural choices to the observed failure modes of the selected baselines. This methodology allows the existing baselines to serve as functional ablation studies by isolating specific algorithmic mechanisms.

\begin{enumerate}
    \item \textbf{Impact of Perimeter-First Strategy:} The necessity of the perimeter extraction mechanism is isolated by comparing SCOPE with the KM ($K^{\text{SCOPE}}$) baseline. Both algorithms operate under the exact same UAV-BS count to ensure a matched hardware budget. However, K-Means relies on centroid-based partitioning, which inherently gravitates toward high-density user centers and neglects sparse boundary users. As demonstrated in Fig.~\ref{fig:satisfaction:a}, SCOPE consistently maintains a significantly higher satisfaction rate by growing clusters from the perimeter inward. This outside-in approach effectively eliminates the service gaps and edge voids that are prevalent in centroid-first heuristics.
    
    \item \textbf{Impact of SEC-driven 3D Altitude Optimization:} The contribution of joint 3D coordinate and altitude optimization is evaluated by comparing SCOPE against the CCS heuristic. While CCS operates with a fixed altitude and a static coverage radius, SCOPE employs Welzl's algorithm to compute the SEC and adaptively selects the optimal altitude $h^*$. The performance gap observed in Fig.~\ref{fig:ee:b} indicates that fixed-parameter methods fail to manage the critical trade-off between interference footprints and link quality. SCOPE's dynamic selection of $h^*$ and $R$ ensures that SINR requirements are satisfied across diverse user clusters, proving that adaptive 3D placement is mandatory for QoS-guaranteed services.

    \item \textbf{Necessity of Tri-Constraint Expansion:} The impact of the constraint-driven expansion mechanism is analyzed by comparing SCOPE with the Voronoi partitioning baseline, which represents an unconstrained geometric partition. As shown in the fairness analysis in Section~\ref{sec:fairness}, the Voronoi method leads to severe load imbalances and capacity saturation because it lacks internal verification for backhaul limits and link feasibility. Conversely, SCOPE's expansion logic incorporates real-time checks for the tri-constraint boundary. This mechanism dynamically limits cluster membership to prevent link congestion, ensuring a highly equitable and feasible distribution of network resources even in stochastic environments.
\end{enumerate}

This ablation analysis confirms that the superior performance of SCOPE is not derived from a single isolated mechanism, but from the synergistic integration of perimeter-first targeting, adaptive 3D optimization, and constraint-aware cluster expansion.

\section{Conclusion}
\label{sec:conclusion}

In this paper, we addressed the fundamental challenge of rapid UAV-BS deployment in heterogeneous post-disaster environments. We proposed the SCOPE framework, which is a deterministic and training-free architecture that synergistically integrates perimeter extraction with the Welzl Smallest Enclosing Circle (SEC) algorithm. Unlike learning-based models that require significant environmental priors and high-end hardware, SCOPE provides a high-fidelity 3D deployment solution in polynomial time. Theoretically, we proved the convergence of the framework and derived a worst-case complexity of $O(N^2 \log N)$, ensuring predictable execution for real-time emergency missions.

Extensive Monte Carlo simulations confirmed the effectiveness of SCOPE across multiple performance dimensions. Under a strictly matched hardware budget, SCOPE achieved a robust and stable user satisfaction rate in high-density scenarios while maintaining sub-second execution latency on standard processing hardware. This level of responsiveness allows for the instantaneous establishment of digital lifelines without the prohibitive training wall associated with DRL methods. Notably, our sensitivity analysis revealed that SCOPE maintains functional service levels even under highly restrictive altitude constraints; in contrast, conventional QoS-blind heuristics suffer a significant performance degradation in such regimes. Furthermore, the tri-constraint pruning mechanism demonstrated a significant load-balancing advantage, yielding a more equitable distribution of network resources even under extreme bottlenecks.

In summary, SCOPE stands out as a robust and agile solution for zero-day disaster response where deployment must be instantaneous and physical constraints are severe. Future research will focus on integrating dynamic trajectory optimization and advanced multi-cell interference coordination to accommodate long-term user mobility and complex energy-efficient flight profiles while maintaining the deterministic reliability established in this framework.

\color{black}

% \section{Conclusion}
% \label{sec:conclusion}

% In this paper, we addressed the challenge of deploying UAV-BSs in environments with highly heterogeneous ground user densities. We proposed a novel 3D deployment framework, SCOPE, which integrates perimeter extraction with the Smallest Enclosing Circle (SEC) algorithm. Unlike traditional counter-clockwise spiral or Voronoi-based methods that rely on fixed-altitude assumptions, SCOPE dynamically adjusts the service radius and altitude of each UAV to fit the irregular shape of user clusters.

% Simulation results using realistic channel models demonstrate that SCOPE effectively overcomes the limitations of existing baselines. In high-density scenarios ($N=1000$), SCOPE achieves up to a twofold increase in user satisfaction compared to the Spiral algorithm. Moreover, extensive analysis confirms that SCOPE offers superior fairness and energy efficiency, making it a robust solution for on-demand 5G/B5G aerial networks. Future work will consider the impact of UAV trajectory planning and dynamic user mobility on the proposed framework.

% --- Put this after Conclusion, before References ---

\section*{Acknowledgment}
The author would like to thank Tzu-Min~Pan, Chang-Lin~Fan, Lun-Hao~Hsu, and Bo-Rui~Chen for their valuable help with the implementation of the baseline algorithms and data collection.

% ----------------------------------------------------

\bibliographystyle{IEEEtran}
\bibliography{reference}
\ifCLASSOPTIONcaptionsoff  \newpage \fi 

\vspace{-20pt}
\begin{IEEEbiography}[{\includegraphics[width=1in,height=1.25in,clip,keepaspectratio]{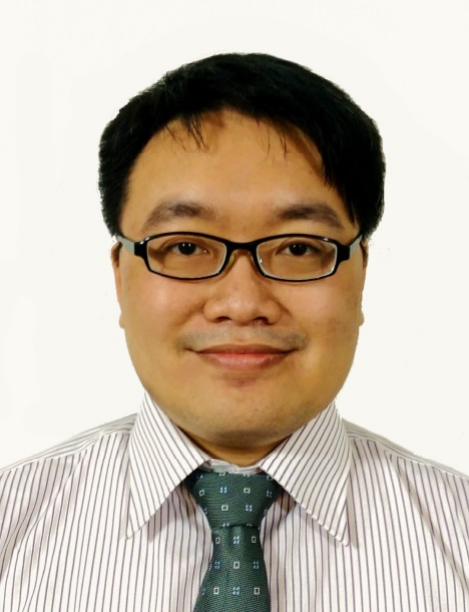}}]{Chuan-Chi Lai}
    (Member, IEEE) received the Ph.D. degree in Computer Science and Information Engineering from National Taipei University of Technology, Taipei, Taiwan, in 2017. He was a postdoctoral research fellow (2017--2019) and contract assistant research fellow (2020) with the Department of Electrical and Computer Engineering, National Chiao Tung University, Hsinchu, Taiwan. From Feb. 2021 to Jan. 2024, he served as an assistant professor at the Department of Information Engineering and Computer Science, Feng Chia University, Taichung, Taiwan. From Feb. 2024, he is currently an assistant professor at the Department of Communications Engineering, National Chung Cheng University, Chiayi, Taiwan. He is a member of the IEEE Vehicular Technology Society and the IEEE Communications Society. His research interests include resource allocation, data management, information dissemination, and distributed query processing for moving objects in emerging applications such as the Internet of Things, edge computing, and next-generation wireless networks. Dr. Lai received the Postdoctoral Researcher Academic Research Award from the Ministry of Science and Technology, Taiwan, in 2019, Best Paper Awards at WOCC 2021 and WOCC 2018, and the Excellent Paper Award at ICUFN 2015.
\end{IEEEbiography}

\end{document}